%% file: main.tex
\documentclass[conference, a4paper]{IEEEtran}
\input{defs}
\IEEEoverridecommandlockouts
\usepackage{cite}
\usepackage{amsmath,amssymb,amsfonts}
\usepackage{algorithmic}
\usepackage{graphicx}
\usepackage{textcomp}
\usepackage{xcolor}
\usepackage[capitalize]{cleveref}
\def\BibTeX{{\rm B\kern-.05em{\sc i\kern-.025em b}\kern-.08em
    T\kern-.1667em\lower.7ex\hbox{E}\kern-.125emX}}
\begin{document}
\author{\IEEEauthorblockN{Cornelia Ott\IEEEauthorrefmark{1}, Hedongliang Liu\IEEEauthorrefmark{2}, Antonia Wachter-Zeh\IEEEauthorrefmark{2}\\}
\IEEEauthorblockA{%
\IEEEauthorrefmark{1}Institute of Communications Engineering, Ulm University, Germany \\
\IEEEauthorrefmark{2}Institute for Communications Engineering, Technical University of Munich (TUM), Germany \\
{\small E-mail: cornelia.ott@uni-ulm.de,  lia.liu@tum.de, antonia.wachter-zeh@tum.de}
}
\thanks{The work of H.~Liu and A.~Wachter-Zeh has been supported by a German Israeli Project Cooperation (DIP) grant under grant no.~PE2398/1-1 and KR3517/9-1.}
}
\title{Geometrical Properties of Balls in Sum-Rank Metric}

\maketitle

\begin{abstract}
The sum-rank metric arises as an algebraic approach for coding in MIMO block-fading channels and multi-shot network coding. Codes designed in the sum-rank metric have raised interest in applications such as streaming codes, robust coded distributed storage systems and post-quantum secure cryptosystems. The sum-rank metric can be seen as a generalization of the well-known Hamming metric and the rank metric. As a relatively new metric, there are still many open theoretical problems for codes in the sum-rank metric. In this paper we investigate the geometrical properties of the balls with sum-rank radii motivated by investigating covering properties of codes.
\end{abstract}


\section{Introduction}
The sum-rank metric arises from the problems in communication over the channels which can be modelled as multiplicative-additive matrix channels, especially with multi-slot usage. 
Before the explicit introduction of the sum-rank metric in multi-shot network
coding literature \cite{nobrega2009multishot,nobrega2010multishot,napp2018concatenation,napp2018multi}, the problem was first considered in coding for MIMO block-fading  channels \cite{el2003design, lu2005unified} and design of AM-PSK constellations \cite{lu2006constructions}. The \emph{minimum sum-rank distance} is a direct analogue to \emph{transmit diversity gain} and the \emph{maximum sum-rank distance} property is a direct analogue to \emph{rate-diversity optimality}. 
An explicit construction of optimal space-time codes from sum-rank metric codes over finite field was first given in \cite{shehadeh2021space}. Other than the studies for communications, sum-rank metric codes have been considered in the applications such as network streaming \cite{mahmood2016convolutional}, distributed storage systems \cite{martinez2019universal,cai2021construction,martinez2020locally} and post-quantum secure code-based cryptosystem \cite{dalconzo2022codeequivalence,hormann2022security}.

Motivated by the various applications, extensive research on sum-rank metric codes has been done in recent years in the subareas of fundamental coding-theoretical properties \cite{byrne2020fundamental,ott2021bounds,martinez2019theory,camps2022optimal,ott2022covering}, constructions of perfect/optimal/systematic sum-rank metric codes \cite{martinez2018skew,alfarano2022sum, martinez2020hamming, almeida2020systematic,martinezpenas2020sumrankBCH,martinez2022generalMSRD, caruso2022duals} and decoding algorithms \cite{boucher2020algorithm,puchinger2021bounds,bartz2021fast,puchinger2022generic,Hoermann2022efficient,Hoermann2022errorerasure,jerkovits2022universal}.

In this paper we present recent results on the geometrical properties of \emph{balls} in the sum-rank metric. A ball in sum-rank metric with center $\x$ and radius $\tau$ is the set of all vectors having sum-rank distance at most $\tau$ to the center $\x$.
The ball around a vector is a fundamental object for investigating several coding-theoretical properties, e.g., the list decoding capacity \cite{puchinger2021bounds}, the sphere-packing bound and Gilbert-Varshamov bound on the size of codes \cite{byrne2020fundamental,ott2021bounds}.

The main contribution of this work is the characterization of the intersection of two balls with sum-rank radii in \cref{sec:volumnIntersection}, which was left as an open problem in the previous study on the covering radius of sum-rank metric codes in \cite{ott2022covering}. We give a general expression on this quantity in \cref{thm:generalVolIntersection} by a simple counting argument and derive simpler expressions for two special cases in \cref{Theo: Intersection case 1} and \cref{Theo: Intersection case 2} respectively by using the tool of elementary linear subspaces.

The rest of the paper is organized as follows: We define the sum-rank metric as well as spheres and balls with sum-rank radii in \cref{sec:pre}. After summarizing known results of the volume of a single ball in the sum-rank metric in \cref{sec:VolBall} and introducing the elementary linear subspace as a tool in \cref{sec:ELS}, we present recent new results on the size of the intersection of two balls.


\section{Preliminaries}
\label{sec:pre}
Let $\Fqm$ be an extension field of a finite field $\Fq$ and let $n, k, \ell, \eta$ be positive integers. In this paper we consider linear codes as $k$ dimensional subspaces of $\Fqm^n$,
where each vector $\x=[\x_1| \ldots| \x_\ell] \in \Fqm^n$ consists of $\ell$ blocks $\x_1, \ldots, \x_\ell \in \Fqm^\eta$ each of length $\eta$. Therefore we have the relation $n=\ell\cdot \eta$. 
Considering $\Fqm$ as an $m$-dimensional vector space over $\Fq$, a vector $\x_i \in \Fqm^\eta$ can also be represented as a matrix $\X_i \in \Fq^{m \times \eta}$.
The rank of $\x_i$ is defined as the rank of the matrix $\X_i$, i.e., $\rkq(\x_i) \coloneqq \rk(\X_i)$. 
For $\x_i \in \Fqm^\eta$ it holds that $\rkq(\x_i)\in \{0, \ldots, \mu\}$, where $\mu\coloneqq \min\{m,\eta\}$. 


\begin{definition}
Let $\x=[\x_1| \ldots| \x_\ell] \in \Fqm^n$. We define the \\\emph{($\ell$-)sum-rank weight} of $\x$ as 
\[
\wtsr: \Fqm^n \rightarrow \NN, \quad \x \mapsto \textstyle\sum_{i=1}^{\ell}\rkq(\x_i)
\]
and the \emph{($\ell$-)sum-rank distance} between two vectors $\x, \x' \in \Fqm^n$ as 
\begin{align*}
\dsr: \Fqm^n \times \Fqm^n &\rightarrow  &&\NN, \\ (\x,\x') &\mapsto  &&\dsr(\x,\x')\coloneqq\wtsr(\x-\x').
\end{align*}
For an arbitrary subspace $\mathcal{V}\subset \Fqm^n$ the ($\ell$-)sum-rank distance of a vector $\x \in \Fqm^n\setminus \mathcal{V}$ to the subspace $\mathcal{V}$ is defined as the ($\ell$-)sum-rank distance of $\x$ to the closest vector in $\mathcal{V}$, i.e.,
\begin{align*}
\dsr(\x,\mathcal{V})=\min_{\v \in \mathcal{V} }\{\dsr(\x,\v)\}.
\end{align*}
\end{definition}

The ($\ell$-)sum-rank distance $\dsr$ is a metric over $\Fqm^n$, the so-called \emph{sum-rank metric}. For $\ell=1$ it corresponds to the rank metric and for $\ell=n$ to the Hamming metric. Hence we denote throughout the paper by $\wtr$, $\wth$, $\dr$ and $\dh$ the weight and the distance in rank metric ($\ell=1$) and in Hamming metric ($\ell=n$), respectively. For a given vector $\x\in \Fqm$ it holds that $\wtr(\x)\leq\wtsr(\x)\leq\wth(\x)$.

Analog to \cite{loidreau2006properties} we define spheres and balls in the sum-rank metric and give definitions for their volume.

\begin{definition}
Let $\tau \in \mathbb{Z}_{\geq 0}$ with $0\leq\tau\leq \ell\cdot \mu$ and $\x \in \Fqm^n$.
The sum-rank metric sphere with radius $\tau$ and center $\x$ is defined as
\[
\Sl(\x,\tau)\coloneqq \{\y\in \Fqm^n \mid \dsr(\x,\y)=\tau\}.
\]
Analogously, we define the ball of sum-rank radius $\tau$ with center $\x$ by
\[
\Bl(\x,\tau) \coloneqq \textstyle\bigcup_{i=0}^{\tau}\Sl(\x,i).
\]
We also define the following cardinalities:
\begin{align*}
\VolS(\tau) &\coloneqq |\Sl(\x,\tau)| ,\\
\VolB(\tau) &\coloneqq|\Bl(\x,\tau)|= \textstyle\sum_{i=0}^{\tau}\VolS(i).
\end{align*}
\end{definition}

The volume of a sphere or a ball is independent of its center, since the sum-rank metric is invariant under translation of vectors,  i.e., $\VolS(\tau)$ and $\VolB(\tau)$ are the volumes of any sphere or ball of radius $\tau$. 
We define the volume of the intersection of two balls with sum-rank radii $\tau_1,\tau_2$ and sum-rank distance $\delta$ between their centers as the number of vectors lying in the intersection, i.e.,
\begin{align*}
    \VolI&(\tau_1,\tau_2,\delta)\\
    \coloneqq & |\Bl(\x_1,\tau_1)\cap\Bl(\x_2, \tau_2)|\ ,
    \text{ with }\dsr(\x_1,\x_2)=\delta\ .
\end{align*}
Note that this quantity is also independent of their centers.
Obviously if $\delta>\tau_1+\tau_2$, $\VolI(\tau_1,\tau_2,\delta)=0$

\section{Balls in the sum-rank metric}
{In this section we give the cardinality of the intersection of two balls in sum-rank metric. In order to do so, we first introduce some known facts about the volume of a single ball in the sum-rank metric in \cref{sec:VolBall} and about the concept of so-called \emph{elementary linear subspaces} in \cref{sec:ELS}. In  \cref{sec:volumnIntersection} we consider the intersection of two balls in the sum-rank metric. In Theorem~\ref{thm:generalVolIntersection} we derive the number of vectors lying in the intersection $\VolI(u,s,t)$ of balls of radii $u$ and $s$, such that their respective centers $\x_1$ and $\x_2$ have sum-rank distance $t$ to each other. Moreover we give this volume of the intersection for two special cases in Theorem~\ref{Theo: Intersection case 1} and Theorem~\ref{Theo: Intersection case 2} for which the computation can be done faster.}

\subsection{Volume of a Single Ball in the Sum-Rank Metric}\label{sec:VolBall}
In order to give this quantity we first introduce some notations and summarize some known results. 
The number of $m\times n$ matrices over $\Fq$ for a given rank $t \leq \min\{m,n\}$ is
\[
\mathrm{NM}_q(n,m,t) \coloneqq \begin{bmatrix}n\\t\end{bmatrix}_{q} \cdot \prod\textstyle_{i=0}^{t-1}(q^m-q^i)
\]
(see e.g., \cite{migler2004weight}), where $\begin{bmatrix}n\\t\end{bmatrix}_{q}\coloneqq\prod_{i=1}^{t}\frac{q^{n-t+i}-1}{q^i-1}$ denotes the $q$-Gaussian binomial coefficient, which is defined by the number of $t$-dimensional subspaces of $\Fq^n$. 
Moreover we define the following set of ordered partitions with bounded number of bounded summands:
\begin{align}
\label{eq:tau_mu}
\tau_{t,\ell,\mu}\coloneqq\left\{\ve{t}=[t_1, \ldots, t_\ell]\mid \textstyle\sum_{i=1}^\ell t_i =t \land t_i\leq \mu, \forall i\right\}\ .
\end{align}
Its cardinality corresponds to the number of possibilities how to partition the sum-rank weight $t$ of a vector into $\ell$ blocks of at most rank $\mu$.
By common combinatorial methods, we obtain (see also \cite[Lemma 1.1]{ratsaby2008estimate})
\begin{align}
\label{ineq: tau}
|\tau_{t,\ell,\mu}| = \textstyle\sum_{i=0}^{\lfloor \frac{t}{\mu+1} \rfloor} (-1)^i \tbinom{\ell}{i}\tbinom{t+\ell-1-(\mu+1)i}{\ell -1} \leq \tbinom{t+\ell-1}{\ell-1}.
\end{align}

If each summand has its own bound, we denote the bounds by a vector $\ve{\mu}=[\mu_1, \ldots, \mu_\ell]$ and define the set of ordered partitions as
\begin{align}
\label{eq:tau_vec_mu}    
\tau_{t,\ell,\ve{\mu}}\coloneqq\left\{\ve{t}=[t_1, \ldots, t_\ell]\mid \textstyle\sum_{i=1}^\ell t_i =t \land t_i\leq \mu_i, \forall i\right\}.
\end{align}

Finally we give the volume of a sphere containing all vectors in $\Fqm^n$ of sum-rank weight $t$, that is $\VolS(t)=\sum_{\ve{t}\in \tau_{t,\ell,\mu}}\prod_{i=1}^\ell \mathrm{NM}_q(\eta,m,t_i)$. Hence the volume of a ball of sum-rank radius $t$ is
$$
\VolB(t)=\sum_{j=0}^{t}\sum_{\ve{t}\in \tau_{j,\ell,\mu}}\prod_{i=1}^\ell \mathrm{NM}_q(\eta,m,t_i).
$$
Note that this can be computed with complexity $\tilde{\mathcal{O}}\big(\ell^2t^3+\ell d^t(m+\eta)\log(q)\big)$  using the efficient algorithm for computing $\VolS$ in \cite[Theorem 6 and Algorithm 1]{puchinger2022generic}.
\subsection{Elementary Linear Subspaces}\label{sec:ELS}
The concept of elementary linear subspaces was extensively studied in \cite{gadouleau2008packing},\cite{gadouleau2009bounds} and \cite{gadouleau2008decoder}. In this subsection we give the definition of an elementary linear subspace, summarize some known results and also adapt a known result for the usage in the sum-rank metric.  
\begin{definition}
Denote by 
$\mathcal{V}\coloneqq\langle \b_1, \ldots, \b_k\rangle_{\Fqm} \subset \Fqm^n$ a $k$-dimensional subspace of $\Fqm$ spanned by the basis vectors $\b_1, \ldots, \b_k$.
If $\b_i\in \Fq^n,\forall i \in \{1, \ldots, k\}$, then $\mathcal{V}$ is called an \emph{elementary linear subspace} of $\Fqm^n$.
Denote the set of all elementary linear subspaces of  of $\Fqm^n$ with $\dim(\mathcal{V})=k$ by $\E$.
\end{definition}

\begin{lemma}{\cite[Lemma 1]{gadouleau2009bounds}}\label{unique ELS}
Each vector $\x\in \Fqm^n$ with $\wtr(\x)=k$ belongs to a unique elementary linear subspace $\mathcal{V}\in \E$.
\end{lemma}

The following lemma draws the connection between elementary linear subspaces of $\Fqm^\eta$ and the sum-rank weight of a vector in $\Fqm^n$. 
\begin{lemma}{\cite[Lemma~5]{ott2022covering}}\label{lem:ELS-sr}
Let $\v=[\v_1|\ldots|\v_\ell] \in\Fqm^n$ with $\v_i \in \Fqm^\eta, \forall i\in\{1, \ldots,\ell\}$  then $\wtsr(\v)\leq k$ if and only if there are elementary linear subspaces $\mathcal{V}_1, \ldots, \mathcal{V}_\ell$ of $\Fqm^\eta$ with $\v_i\in\mathcal{V}_i, \forall i\in\{1, \ldots,\ell\}$ such that $\sum_{i=1}^{\ell}\dim(\mathcal{V}_i)=k$.
\end{lemma}
\begin{lemma}\label{unique cart prod}
Let $\x=[\x_1|\ldots|\x_\ell]\in\Fqm^n$ be a vector of sum-rank weight $\tau$ with a fix sum-rank weight distribution $\rk(\x_i)=\tau_i$ for all $i\in \{1, \ldots, \ell\}$ and $\sum_{i=1}^{\ell}\tau_i=\tau$.
Then there exists a unique cartesian product of elementary linear subspaces $\mathcal{V}\coloneqq\mathcal{V}_1\times\ldots\times\mathcal{V}_n$ with $\mathcal{V}_i\in\mathcal{E}_{\tau_i}(\Fqm^\eta)$ such that $\x\in \mathcal{V}$.
\end{lemma}
\begin{proof}
Follows directly by Lemma~\ref{unique ELS}.
\end{proof}
For any $\cV\in \E$ there is a $\Bar{\mathcal{V}}\in \mathcal{E}_{n-k}(\Fqm^n)$ such that $\cV\oplus\Bar{\mathcal{V}}=\Fqm^n$.
For any vector $\v\in \Fqm^n$ we denote by $\v^{(\mathcal{V})}$ the projection of $\v$ on $\cV$ along $\Bar{\mathcal{V}}$.
\begin{lemma}{\cite[Lemma 3]{gadouleau2008packing}}
\label{Gadouleau Lemma 3}
Let $\mathcal{V}\in\E$ and let $\u\in \mathcal{V}$ having rank $k$, then $\rk(\u^{(\mathcal{A})})=a$ and $\rk(\u^{(\mathcal{B})})=k-a$ for any $\mathcal{A}\in \mathcal{E}_a(\Fqm^n)$ and $\mathcal{B}\in\mathcal{E}_{k-a}(\Fqm^n)$ such that $\mathcal{A}\oplus\mathcal{B}=\mathcal{V}$. 
\end{lemma}
\begin{lemma}{\cite[Lemma 4]{gadouleau2008packing}}
\label{Gadouleau Lemma 4}
Let $\mathcal{V}\in\E$ and let $\u\in \mathcal{V}$ having rank $k$. For any $\mathcal{A}\in \mathcal{E}_a(\Fqm^n)$ and $\mathcal{B}\in\mathcal{E}_{k-a}(\Fqm^n)$ such that $\mathcal{A}\oplus\mathcal{B}=\cV$, the functions 
$$\phi_{\u}: \mathcal{E}_a(\Fqm^n)\times \mathcal{E}_{k-a}(\Fqm^n)\rightarrow \mathcal{A}, \quad(\mathcal{A}, \mathcal{B})\mapsto \u^{(\mathcal{A})}$$
$$\psi_{\u}: \mathcal{E}_a(\Fqm^n)\times \mathcal{E}_{k-a}(\Fqm^n)\rightarrow \mathcal{B}, \quad(\mathcal{A}, \mathcal{B})\mapsto \u^{(\mathcal{B})}$$
are both injective.
\end{lemma}
\begin{lemma}{\cite[Lemma 2]{gadouleau2008packing}}
\label{Gadouleau Lemma 2}
Let $\mathcal{V}\in\E$ and let $\mathcal{A}\in \mathcal{E}_a(\mathcal{V})$. Then there exist $q^{a\cdot(k-a)}$ elementary linear subspaces $\mathcal{B}\in \mathcal{E}_{k-a}(\mathcal{V})$ such that $\mathcal{A}\oplus\mathcal{B}=\mathcal{V}$. Moreover there are $q^{a\cdot(k-a)}\cdot\begin{bmatrix}k\\a\end{bmatrix}_{q} $ such ordered pairs $(\mathcal{A},\mathcal{B})$.
\end{lemma}
\subsection{Volume of the Intersection of Two Balls in the Sum-Rank Metric}
\label{sec:volumnIntersection}
In this section we give the number of vectors lying in the intersection of two sum-rank metric balls. First we give a general expression on this cardinality in \cref{thm:generalVolIntersection} and then 
derive the expressions for two special cases in \cref{Theo: Intersection case 1} and \cref{Theo: Intersection case 2} respectively, which can be computed faster.

In \cite{gadouleau2009bounds} the number of vectors lying in the intersection of two spheres in the rank metric of radii $u$ and $s$ and distance $t$ between their centers was derived, that is
\begin{align*}
  &\mathcal{J}(u,s,t,n,m)\\
  &\coloneqq\frac{\sum_{i=0}^n \mathrm{NM}_q(n,m,i)\mathcal{K}_u(i,n,m)\mathcal{K}_s(i,n,m)\mathcal{K}_t(i,n,m)}{q^{mn}\mathrm{NM}_q(n,m,t)},   
\end{align*}
where $\mathcal{K}_j(i,n,m)$ is a $q$-Krawtchouk polynomial (see \cite{delsarte1976properties}) and defined as 
$$
\mathcal{K}_j(i,n,m)\coloneqq\sum_{l=0}^j (-1)^{j-l}q^{lm+\binom{j-l}{2}}\begin{bmatrix}n-l\\n-j\end{bmatrix}_q\begin{bmatrix}n-i\\l\end{bmatrix}_q.
$$
Moreover the cardinality of all vectors lying in the intersection of two balls with rank metric radii $u$ and $s$ and distance $t$ between their centers was given in \cite{gadouleau2009bounds},
$$
\mathcal{I}(u,s,t,n,m)\coloneqq \sum_{i=0}^u \sum_{j=0}^s \mathcal{J}(u,s,t,n,m).
$$
This leads to the volume of the intersection of two balls in sum-rank metric.
\begin{theorem}\label{thm:generalVolIntersection}
Let $u$, $s$, $t$ be positive integers such that $u+s\geq t$. The number of vectors $\ve{v}\in \Fqm^n$ lying in the intersection of two balls with sum-rank radii $u$ and $s$ and sum-rank distance $t$ between their centers is
\begin{align*}
& \VolI(u,s,t)=\\
&\sum_{\begin{subarray}[\ve{u}=[u_1, \ldots, u_\ell]\\ \in \tau_{u,\ell,\mu}\end{subarray}}
\sum_{\begin{subarray}[\ve{s}=[s_1, \ldots, s_\ell]\\ \in \tau_{s,\ell,\mu}\end{subarray}}
\sum_{\begin{subarray}[\ve{t}=[t_1, \ldots, t_\ell]\\ \in \tau_{t,\ell,\mu}\end{subarray}}
\prod_{i=1}^\ell \mathcal{I}(u_i,s_i,t_i,\eta,m).
\end{align*}
\end{theorem}
\begin{proof}
Let $\Bl(\x,u)$ and $\Bl(\y,s)$ be two sum-rank metric balls in $\Fqm^n$ with $\dsr(\x,\y)=t$.
For fix partitions of their radii $u=\sum_{i=1}^\ell u_i$ and $s=\sum_{i=1}^\ell s_i$ and the distance between their centers $t=\sum_{i=1}^\ell t_i$, we have that for each block $i\in 1, \ldots, \ell$ the number of vectors $\ve{v}_i\in \Fqm^\eta$ with at most  rank distance $u_i$ to $\x_i$ and also $\dr(\y_i,\ve{v}_i)\leq s_i$ is given by
$\mathcal{I}(u_i,s_i,t_i,\eta,m)$. Hence the number of vectors $\ve{v}\in \Fqm^n$ with at most  sum-rank distance $u$ to $\x$ and  sum-rank distance at most $s$ to $\y$ is given by $\prod_{i=1}^\ell \mathcal{I}(u_i,s_i,t_i,\eta,m)$
for this choice of the partitions. For an arbitrary choice of the partitions of the radii $u$ and $s$ and the distance between the centers $t$ the claim follows. 
\end{proof}
Computing the cardinality of the intersection of two spheres for this general case is computationally demanding. Therefore, in the following we give the volume of the intersection of two spheres for two special cases, which can be computed efficiently.
\begin{lemma} \label{lemma:intersection}
Let $\x,\y\in\Fqm^n$ with $\wtr(\y)=1$ and $\wtr(\x)=r$.
Then there exist $$\frac{(q^n-q^r)(q^m-q^r)}{q-1}$$ such pairs of $(\x,\y)$ fulfilling $\wtr(\x)+\wtr(\y)=\wtr(\x-\y)$.
\end{lemma}
\begin{proof}
Let $\x,\y\in\Fqm^n$ with $\wtr(\y)=1$ and $\wtr(\x)=r$. 
We denote by $\rowSpan{\x}$ the $\Fq$-row space of $\X$ (the matrix representation in $\Fq^{m\times n}$ of $\x\in\Fqm^n$) and by $\colSpan{\x}$ the $\Fq$-column space of $\X$.
Then  $\wtr(\x)+\wtr(\y)=\wtr(\x-\y)$ if and only if $\dim(\rowSpan{\x}\cap\rowSpan{\y})=\dim(\colSpan{\x}\cap\colSpan{\y})=\{\0\}$ whereas this statement is equivalent to $\rowSpan{\y}\not\subset\rowSpan{\x}$ and $\colSpan{\y}\not\subset\colSpan{\x}$. The probability that $\rowSpan{\y}\subset\rowSpan{\x}$ under the condition that $\y$ has rank $1$ is given by 
\begin{align*}
  P(\rowSpan{\y}\subset\rowSpan{\x}\mid \wtr(y)=1)=\frac{\begin{bmatrix}r\\1\end{bmatrix}_{q} }{\begin{bmatrix}n\\1\end{bmatrix}_{q}}=\frac{q^r-1}{q^n-1} 
\end{align*}
and hence \begin{align*}
  P(\rowSpan{\y}\not\subset\rowSpan{\x}\mid \wtr(y)=1)=1-\frac{q^r-1}{q^n-1}.
\end{align*}
A similar statement can be obtained for the corresponding column space:
\begin{align*}
  P(\colSpan{\y}\not\subset\colSpan{\x}\mid \wtr(y)=1)=1-\frac{q^r-1}{q^m-1}.
\end{align*}
Since $|\{\y\in \Fqm^n\mid \wtr(\y)=1\}|=\begin{bmatrix}n\\1\end{bmatrix}_{q}\cdot(q^m-1)=\frac{(q^m-1)(q^n-1)}{q-1}$, one gets
 \begin{align*}
  &P(\rowSpan{\y}\not\subset\rowSpan{\x} \land\colSpan{\y}\not\subset\colSpan{\x} \mid \wtr(y)=1) \\&\cdot |\{\y\in \Fqm^n\mid \wtr(\y)=1\}|\\
  &=\frac{(q^m-1)(q^n-1)}{q-1}\cdot\Big(1-\frac{q^r-1}{q^n-1}\Big)\cdot\Big(1-\frac{q^r-1}{q^m-1}\Big)\\
  &=\frac{(q^n-q^r)(q^m-q^r)}{q-1}.
 \end{align*}
\end{proof}
We show in the following theorems the number of the vectors lying in the intersection of two balls with sum-rank metric radii for two special cases. Analog statements for the rank metric were given in \cite[Proposition 4 and 5]{gadouleau2008packing}.
\begin{theorem}
\label{Theo: Intersection case 1}
Let $\x, \y\in  \Fqm^n$ such that $\dsr(\x,\y)=\delta$. Then 
\begin{align*}
    \VolI(\delta,1,\delta)=& |\Bl(\x, \delta)\cap \Bl(\y, 1)|\\
    =& 1+\frac{(q^m-1)(q^n-1)}{q-1}\\
    &-\sum_{\begin{subarray}[\ve{\delta}=[\delta_1, \ldots, \delta_\ell]\\\in \tau_{\delta,\ell,\mu}\end{subarray}}\sum_{i=1}^\ell \frac{(q^\eta-q^{\delta_i})\cdot(q^m-q^{\delta_i})}{q-1}\ .
\end{align*}
\end{theorem}
\begin{proof}
 Let $\x, \y\in  \Fqm^n$ such that $\dsr(\x,\y)=\delta$. Since the sum-rank metric is invariant under the translation of vectors we assume w.l.o.g. that $\y = \0$ and $\wtsr(\x)=\delta$.
 The set of vectors lying in the intersection $\Bl(\x, \delta)\cap \Bl(\y, 1)$ consist of the zero vector 
 as well as all vectors  $\v\in\Fqm^n$ lying on the sphere $\Sl(\y,1)$ without those having sum-rank distance at least $\delta+1$ from $\x$. However since 
 $$\dsr(\v,\x)\leq\wtsr(\x)+\wtsr(\v)=\delta+1$$ 
 we only need to subtract the number of vectors having exactly  sum-rank distance $\delta+1$ to $\x$. Therefore  
 \begin{align*}
 &|\Bl(\x, \delta)\cap \Bl(\y, 1)|=|\{\0\}|+\\
 &|\{\v\in\Sl(\y,1)\}|-|\{\v\in\Sl(\y,1)\mid\dsr(\v,\x)=\delta+1\}|.    
 \end{align*}
 Since the sum-rank weight of a vector is $1$ if and only if its rank weight is $1$, the cardinality 
 \begin{align*}
   &|\{\v\in\Sl(\y,1)\}|=|\{\v\in \Fqm^n\mid \wtr(\v)=1\}|\\
   =& \frac{(q^m-1)(q^n-1)}{q-1}.  
 \end{align*}
 Now in order to compute $$\{\v\in\Sl(\y,1)\} \setminus\{\v\in\Sl(\y,1)\mid\dsr(\v,\x)=\delta+1\}$$
 we consider a fixed sum-rank weight decomposition of $\x=[\x_1|\ldots|\x_\ell]$ such that  $\rk(\x_i)=\delta_i$ for all $i\in \{1, \ldots, \ell\}$ and $\sum_{i=1}^{\ell}\delta_i=\delta$. 
Now let $\v\in \Fqm^n$ be a non-zero vector with $\wtsr(\v)=1$, i.e., there is exactly on $i\in \{1, \ldots, \ell\}$ such that $\rk(\v_i)=1$ and $\rk(\v_j)=0$  for all $j\in \{1, \ldots,i-1,i+1,\ldots, n\}$. Therefore for a fix non-zero block $\v_i$ it holds, that $\dsr(\v,\x)= \delta+1$ if and only if $\dr(\v_i,\x_i)=\delta_i+1$. 
Now we want to find the number of vectors $\v_i\in \Fqm^\eta$ having rank distance $\delta_i +1$ to the vector $\x_i$, which has rank $\delta_i$. By lemma~\ref{lemma:intersection} there are  $\frac{(q^\eta-q^{\delta_i})\cdot(q^m-q^{\delta_i})}{q-1}$ such vectors.
We have this amount of vectors for each possible non-zero block $i\in \{1, \ldots, \ell\}$, so in total there are $\sum_{i=1}^\ell \frac{(q^\eta-q^{\delta_i})\cdot (q^m-q^{\delta_i})}{q-1}$ such vector $\v$ in the intersection for a fix partition of $\delta$ and hence the claim follows.
\end{proof}

\begin{theorem}
\label{Theo: Intersection case 2}
Let $\x, \y\in  \Fqm^n$ such that $\dsr(\x,\y)=\delta$ then 
\begin{align*}
    &\VolI(\gamma,\delta-\gamma,\delta)=|\Bl(\x, \gamma)\cap \Bl(\y, \delta-\gamma)|\\
    &=\sum_{\begin{subarray} [\ve{\delta}=[\delta_1, \ldots, \delta_\ell]\\ \in \tau_{\delta,\ell,\mu}\end{subarray}}\sum_{\begin{subarray}[\ve{\gamma}=[\gamma_1, \ldots, \gamma_\ell]\\ \in \tau_{\gamma,\ell,\ve{\delta}}\end{subarray}}\sum_{i=1}^\ell q^{\gamma_i\cdot(\delta_i-\gamma_i)}\cdot\begin{bmatrix}\delta_i\\\gamma_i\end{bmatrix}_{q}
\end{align*}
for $0\leq\gamma\leq\delta$. 
\end{theorem}
\begin{proof}
Let $\x, \y\in  \Fqm^n$ such that $\dsr(\x,\y)=\delta$. Because of the translational symmetry of the sum-rank metric we again assume w.l.o.g. that $\x = \0$ and $\wtsr(\y)=\delta$. Let's fix a partition of $\wtsr(\y)=\delta=\sum_{i=1}^\ell\delta_i$, then from Lemma~\ref{unique cart prod} it follows that there is a unique cartesian product of elementary linear subspaces $\mathcal{V}\coloneqq\mathcal{V}_1\times\ldots\times\mathcal{V}_n$ with $\mathcal{V}_i\in\mathcal{E}_{\delta_i}(\Fqm^\eta)$ such that $\y\in \mathcal{V}$.
We also fix a partition of $\gamma=\sum_{i=1}^\ell\gamma_i$.

We first prove that all vectors $\v\in\Bl(\0, \gamma)\cap \Bl(\y, \delta-\gamma)$ are lying in $\cV=\mathcal{V}_1\times\ldots\times\mathcal{V}_n$. 
Let $\v_i=\v_i^{\mathcal{V}_i}+\v_i^{\mathcal{W}_i}$, where  $\mathcal{W}_i\in \mathcal{E}_{\eta-\delta_i}(\Fqm^\eta)$ such that $\mathcal{V}_i\oplus \mathcal{W}_i=\Fqm^\eta$ for each block $i \in \{1, \ldots, \ell\}$ and $\mathcal{W}\coloneqq\mathcal{W}_1\times\ldots\times \mathcal{W}_\ell$. Now for all  $\v\in\Bl(\0, \gamma)\cap \Bl(\y, \delta-\gamma)$ we have that 
\begin{align}
\label{wt(v) leq gamma}
    \wtsr(\v^{\mathcal{V}})\leq \wtsr(\v)\leq \gamma,
\end{align}
since $\v\in\Bl(\0,\gamma)$ and $\v\in\Bl(\y,\delta-\gamma)$, we have
  \begin{align}
  \label{wt(y-v) leq delta-gamma}
      \wtsr((\y-\v)^{\mathcal{V}})\leq \wtsr(\y-\v)\leq \delta-\gamma\ .
  \end{align}
Since $\v^{\mathcal{V}}+(\y-\v)^{\mathcal{V}}=(\v+(\y-\v))^\cV=\y^{(\mathcal{V})}=\y$,
we have
\begin{align*}
\delta=\wtsr(\y)&=\wtsr(\v^\cV+(\y-\v)^\cV)\\
&\leq\wtsr(\v^\cV)+\wtsr((\y-\v)^\cV)\\
&\leq \gamma+\delta-\gamma =\delta\ .
\end{align*}
Then together with the inequalities in \eqref{wt(v) leq gamma} and \eqref{wt(y-v) leq delta-gamma}, we have that
\begin{align}
\label{wt(v) = gamma}
    \wtsr(\v^{\mathcal{V}})=& \wtsr(\v)= \gamma\\
 \label{wt(y-v) = delta-gamma}
    \wtsr((\y-\v)^{\mathcal{V}})=& \wtsr(\y-\v)= \delta-\gamma. 
\end{align}
Now \eqref{wt(v) = gamma} leads to 
\begin{align}
\label{eq:sum rk(vi)}
    \gamma=\sum_{i=1}^\ell\rkq(\v_i^{\mathcal{V}})= \sum_{i=1}^\ell \rkq(\v_i)=\sum_{i=1}^\ell \rkq(\v_i^{\mathcal{V}}+\v_i^{\mathcal{W}})\ .
\end{align}
Since $\v_i^{\mathcal{V}}$ and $\v_i^{\mathcal{W}}$ must be linearly independent (because $\mathcal{V}_i\oplus\mathcal{W}_i=\Fqm^\eta$) it follows that 
$\rkq(\v_i^{\mathcal{V}}) \leq \rkq(\v_i^{\mathcal{V}}+\v_i^{\mathcal{W}})= \rkq(\v_i)$ for all $i\in \{1, \ldots, \ell\}$.
Together with equation~\eqref{eq:sum rk(vi)} it implies that 
\begin{align}
\label{eq:rk(vi^V)=rk(vi)}
  \rkq(\v_i^{\mathcal{V}}) = \rkq(\v_i)\ ,\ \forall i\in \{1, \ldots, \ell\}\ . 
\end{align}
Analogously we get from equation~\eqref{wt(y-v) = delta-gamma} that 
\begin{align}
\label{eq:rk((y_i-vi)^V)=rk(y_i-vi)}
   \rkq((\y_i-\v_i)^{\mathcal{V}}) = \rkq(\y_i-\v_i)\ ,\ \forall i\in \{1, \ldots, \ell\}\ .
\end{align}
Moreover $\langle\v^{(\mathcal{V})}\rangle_{\Fqm} \cap \langle(\y-\v)^{(\mathcal{V})}\rangle_{\Fqm}=\{\0\}$. 
In particular 
\begin{align}
\label{<vi^Vi>cap<(y_i-v_i)^Vi>=0}
  \langle\v_i^{(\mathcal{V}_i)}\rangle_{\Fqm} \cap \langle(\y_i-\v_i)^{(\mathcal{V}_i)}\rangle_{\Fqm}=\{\0\}\ ,\ \forall i\in \{1, \ldots, \ell\}\ .  
\end{align}
Equation~\eqref{eq:rk(vi^V)=rk(vi)} implies that 
\begin{align}
\label{rel: yi^Wi subset y_i^vi}
    \langle\v_i^{(\mathcal{W}_i)}\rangle_{\Fqm}\subset\langle\v_i^{(\mathcal{V}_i)}\rangle_{\Fqm}\ , \forall i\in \{1, \ldots, \ell\}
\end{align}
and similarly it follows with equation~\eqref{eq:rk((y_i-vi)^V)=rk(y_i-vi)} that
\begin{align}
\label{rel: (vi-yi)^Wi subset (vi-y_i)^vi}
  \langle(\y_i-\v_i)^{(\mathcal{W}_i)}\rangle_{\Fqm}\subset\langle(\y_i-\v_i)^{(\mathcal{V}_i)}\rangle_{\Fqm}\ ,\ \forall i\in \{1, \ldots, \ell\}\ .
\end{align}
Now with the relations \eqref{<vi^Vi>cap<(y_i-v_i)^Vi>=0}, \eqref{rel: yi^Wi subset y_i^vi} and \eqref{rel: (vi-yi)^Wi subset (vi-y_i)^vi} we obtain
\begin{align}
    \langle\v_i^{(\mathcal{W}_i)}\rangle_{\Fqm} \cap \langle(\y_i-\v_i)^{(\mathcal{W}_i)}\rangle_{\Fqm}=\{\0\}\ ,\ \forall i\in \{1, \ldots, \ell\}\ .
\end{align}
Together with $\y_i^{(\mathcal{W}_i)}+(\v_i-\y_i)^{(\mathcal{W}_i)}=\0$ for all $i\in \{1, \ldots, \ell\}$ one gets
$\v_i^{(\mathcal{W}_i)}=(\y_i-\v_i)^{(\mathcal{W}_i)}=\0$  and hence $\v_i\in \mathcal{V}_i$ for all $i\in \{1, \ldots, \ell\}$ which implies that $\v\in \mathcal{V}$.

In the second part of this proof we show that $\v_i$ is necessarily a projection of $\y_i$ onto some elementary linear subspace $\mathcal{A}_i$ of $\mathcal{V}_i$ for all $i \in \{1, \ldots, \ell\}$.
If $\v_i\in \mathcal{V}_i$ fulfills that $\rkq(\v_i)=\gamma_i$ and $\rkq(\y_i-\v_i)=\delta_i-\gamma_i$ for all $i \in \{1, \ldots, \ell\}$ then each $\v_i$ belongs to some elementary linear subspace $\mathcal{A}_i$ of $\mathcal{V}_i$ and each $\y_i-\v_i$ belongs to some elementary linear subspace $\mathcal{B}_i$ of $\mathcal{V}_i$ such that $\mathcal{A}_i\oplus\mathcal{B}_i=\mathcal{V}_i$ for all $i \in \{1, \ldots, \ell\}$. Hence $\v_i=\y_i^{(\mathcal{A})}$ and $\y_i-\v_i=\y_i^{(\mathcal{B})}$ for all $i \in \{1, \ldots, \ell\}$.
Conversely for any $\mathcal{A}_i\in \mathcal{E}_ {\gamma_i }(\Fqm^\eta)$ and any $\mathcal{B}_i\in \mathcal{E}_{\delta_i-\gamma_i}(\Fqm^\eta)$ such that $\mathcal{A}_i\oplus\mathcal{B}_i=\cV_i$,  the vector $\y_i^{(\mathcal{A})}$ has rank weight $\gamma_i$ and rank distance $\delta_i-\gamma_i$ from $\y_i$ by Lemma~\ref{Gadouleau Lemma 3}. From Lemma~\ref{Gadouleau Lemma 4} it follows that all these vectors $\y_i^{(\mathcal{A})}$ are distinct. 
Hence for each block we have as many vectors $\y_i$ as ordered pairs $(\mathcal{A},\mathcal{B})$ for a fix partition of $\delta$ and $\gamma$. It follows from Lemma~\ref{Gadouleau Lemma 2} that there are $q^{\gamma_i\cdot(\delta_i-\gamma_i)}\cdot\begin{bmatrix}\delta_i\\\gamma_i\end{bmatrix}_{q}$ such ordered pairs. And hence there are $\sum_{i=1}^\ell q^{\gamma_i\cdot(\delta_i-\gamma_i)}\cdot\begin{bmatrix}\delta_i\\\gamma_i\end{bmatrix}_{q}$ vectors $\v=[\v_1|\ldots|\v_n]$ such that every block $\v_i$ is projection of $\y_i$ onto some elementary linear subspace $\mathcal{A}_i$ of $\mathcal{V}_i$ for a fix partition of $\delta$ and $\gamma$. Finally we get 
\begin{align*}
    &|\Bl(\x, \gamma)\cap \Bl(\y, \delta-\gamma)|\\
    =&\sum_{\begin{subarray} [\ve{\delta}=[\delta_1, \ldots, \delta_\ell]\\ \in \tau_{\delta,\ell,\mu}\end{subarray}}\sum_{\begin{subarray}[\ve{\gamma}=[\gamma_1, \ldots, \gamma_\ell]\\ \in \tau_{\gamma,\ell,\ve{\delta}}\end{subarray}}\sum_{i=1}^\ell q^{\gamma_i\cdot(\delta_i-\gamma_i)}\cdot\begin{bmatrix}\delta_i\\\gamma_i\end{bmatrix}_{q}\ ,
\end{align*}
where $\tau_{\delta,\ell,\mu}$ and $\tau_{\gamma,\ell,\ve{\delta}}$ are defined as in \eqref{eq:tau_mu} and \eqref{eq:tau_vec_mu}, respectively.
\end{proof}
\section{Conclusion and Future Work}
Motivated by an open problem from a previous work \cite{ott2022covering} on covering properties of sum-rank metric codes, we derived the volume of the intersection of two balls of sum-rank radii $u$ and $s$ having sum-rank distance $t$ between their respective centers. Furthermore we considered two special cases where this volume can be computed efficiently and we derived closed expressions for the volume of the intersection of two balls in these cases.
Finding upper and lower bounds on the volume of the intersection
of two sum-rank metric balls that can be computed fast would also be interesting and helpful for considering covering properties of sum-rank metric codes. 

\bibliographystyle{IEEEtran}
\bibliography{refs}

\end{document}

%% file: defs.tex
\usepackage{lipsum}
\usepackage{amsmath,amssymb,amsthm}
\usepackage{cite}
\usepackage{tikz}
\usepackage{pgfplots}
\usepackage{tikzscale}
\pgfplotsset{compat=newest}
\usepackage{xcolor}
\usepackage{mathtools}
\usepackage{url}
\usepackage{booktabs}
\usepackage{enumitem}

\newtheorem{theorem}{Theorem}
\newtheorem{definition}{Definition}

\newtheorem{lemma}{Lemma}

\def\ve#1{{\mathchoice{\mbox{\boldmath$\displaystyle #1$}}%
		{\mbox{\boldmath$\textstyle #1$}}%
		{\mbox{\boldmath$\scriptstyle #1$}}%
		{\mbox{\boldmath$\scriptscriptstyle #1$}}}}

\newcommand{\X}{\ve{X}}

\newcommand{\0}{\ve{0}}

\renewcommand{\b}{\ve{b}}

\renewcommand{\u}{\ve{u}}
\renewcommand{\v}{\ve{v}}

\newcommand{\x}{\ve{x}}
\newcommand{\y}{\ve{y}}

\newcommand{\cV}{\mathcal{V}}

\newcommand{\Fq}{\mathbb{F}_q}
\newcommand{\Fqm}{\mathbb{F}_{q^m}}
\newcommand{\NN}{\mathbb{N}}

\newcommand{\E}{\mathcal{E}_k(\Fqm^n)}
\newcommand{\wth}{\mathrm{wt}_{SR,n}}
\newcommand{\wtr}{\mathrm{wt}_{SR,1}}

\newcommand{\wtsr}{\mathrm{wt}_{SR,\ell}}
\newcommand{\dsr}{\mathrm{d}_{SR,\ell}}
\newcommand{\dr}{\mathrm{d}_{SR,1}}
\renewcommand{\dh}{\mathrm{d}_{SR,n}}

\newcommand{\rk}{\mathrm{rk}}
\newcommand{\rkq}{\mathrm{rk}_q}

\newcommand{\VolS}{\mathrm{Vol}_{\mathcal{S}_{\ell}}}
\newcommand{\VolB}{\mathrm{Vol}_{\mathcal{B}_{\ell}}}
\newcommand{\VolI}{\mathrm{Vol}_{\mathcal{I}_{\ell}}}

\newcommand{\Sl}{\mathcal{S}_\ell}
\newcommand{\Bl}{\mathcal{B}_\ell}

\newcommand{\rowSpan}[1]{\left\langle #1 \right\rangle_{\mathrm{row}}}
\newcommand{\colSpan}[1]{\left\langle #1 \right\rangle_{\mathrm{col}}}

%% file: main.bbl
\begin{thebibliography}{10}
\providecommand{\url}[1]{#1}
\csname url@samestyle\endcsname
\providecommand{\newblock}{\relax}
\providecommand{\bibinfo}[2]{#2}
\providecommand{\BIBentrySTDinterwordspacing}{\spaceskip=0pt\relax}
\providecommand{\BIBentryALTinterwordstretchfactor}{4}
\providecommand{\BIBentryALTinterwordspacing}{\spaceskip=\fontdimen2\font plus
\BIBentryALTinterwordstretchfactor\fontdimen3\font minus
  \fontdimen4\font\relax}
\providecommand{\BIBforeignlanguage}[2]{{%
\expandafter\ifx\csname l@#1\endcsname\relax
\typeout{** WARNING: IEEEtran.bst: No hyphenation pattern has been}%
\typeout{** loaded for the language `#1'. Using the pattern for}%
\typeout{** the default language instead.}%
\else
\language=\csname l@#1\endcsname
\fi
#2}}
\providecommand{\BIBdecl}{\relax}
\BIBdecl

\bibitem{nobrega2009multishot}
R.~W. N{\'o}brega and B.~F. Uch{\^o}a-Filho, ``{Multishot codes for network
  coding: Bounds and a multilevel construction},'' in \emph{2009 IEEE
  International Symposium on Information Theory}.\hskip 1em plus 0.5em minus
  0.4em\relax IEEE, 2009, pp. 428--432.

\bibitem{nobrega2010multishot}
------, ``{Multishot codes for network coding using rank-metric codes},'' in
  \emph{2010 Third IEEE International Workshop on Wireless Network
  Coding}.\hskip 1em plus 0.5em minus 0.4em\relax IEEE, 2010, pp. 1--6.

\bibitem{napp2018concatenation}
D.~Napp, R.~Pinto, and V.~Sidorenko, ``{Concatenation of convolutional codes
  and rank metric codes for multi-shot network coding},'' \emph{Designs, Codes
  and Cryptography}, vol.~86, no.~2, pp. 303--318, 2018.

\bibitem{napp2018multi}
D.~Napp and F.~Santana, ``{Multi-shot network coding},'' in \emph{Network
  Coding and Subspace Designs}.\hskip 1em plus 0.5em minus 0.4em\relax
  Springer, 2018, pp. 91--104.

\bibitem{el2003design}
H.~El~Gamal and A.~R. Hammons, ``{On the design of algebraic space-time codes
  for MIMO block-fading channels},'' \emph{IEEE Transactions on Information
  Theory}, vol.~49, no.~1, pp. 151--163, 2003.

\bibitem{lu2005unified}
H.-f. Lu and P.~V. Kumar, ``{A unified construction of space-time codes with
  optimal rate-diversity tradeoff},'' \emph{IEEE Transactions on Information
  Theory}, vol.~51, no.~5, pp. 1709--1730, 2005.

\bibitem{lu2006constructions}
H.-F. Lu, ``{On constructions of algebraic space-time codes with AM-PSK
  constellations satisfying rate-diversity tradeoff},'' \emph{IEEE transactions
  on information theory}, vol.~52, no.~7, pp. 3198--3209, 2006.

\bibitem{shehadeh2021space}
M.~Shehadeh and F.~R. Kschischang, ``{Space--time codes from sum-rank codes},''
  \emph{IEEE Transactions on Information Theory}, vol.~68, no.~3, pp.
  1614--1637, 2021.

\bibitem{mahmood2016convolutional}
R.~Mahmood, A.~Badr, and A.~Khisti, ``{Convolutional codes with maximum column
  sum rank for network streaming},'' \emph{IEEE Transactions on Information
  Theory}, vol.~62, no.~6, pp. 3039--3052, 2016.

\bibitem{martinez2019universal}
U.~Mart{\'\i}nez-Pe{\~n}as and F.~R. Kschischang, ``{Universal and dynamic
  locally repairable codes with maximal recoverability via sum-rank codes},''
  \emph{IEEE Transactions on Information Theory}, 2019.

\bibitem{cai2021construction}
H.~Cai, Y.~Miao, M.~Schwartz, and X.~Tang, ``{A construction of maximally
  recoverable codes with order-optimal field size},'' \emph{IEEE Transactions
  on Information Theory}, vol.~68, no.~1, pp. 204--212, 2021.

\bibitem{martinez2020locally}
U.~Mart{\'\i}nez-Pe{\~n}as and D.~Napp, ``{Locally repairable convolutional
  codes with sliding window repair},'' \emph{IEEE Transactions on Information
  Theory}, vol.~66, no.~8, pp. 4935--4947, 2020.

\bibitem{dalconzo2022codeequivalence}
G.~D'Alconzo, ``{Code Equivalence in the Sum-Rank Metric: Hardness and
  Completeness},'' \emph{Cryptology ePrint Archive}, 2022.

\bibitem{hormann2022security}
``{}security considerations for mceliece-like cryptosystems based on linearized
  reed-solomon codes in the sum-rank metric.''

\bibitem{byrne2020fundamental}
E.~Byrne, H.~Gluesing-Luerssen, and A.~Ravagnani, ``{Fundamental Properties of
  Sum-Rank-Metric Codes},'' \emph{IEEE Transactions on Information Theory},
  vol.~67, no.~10, pp. 6456--6475, 2021.

\bibitem{ott2021bounds}
C.~Ott, S.~Puchinger, and M.~Bossert, ``{Bounds and genericity of
  sum-rank-metric codes},'' in \emph{2021 XVII International Symposium"
  Problems of Redundancy in Information and Control
  Systems"(REDUNDANCY)}.\hskip 1em plus 0.5em minus 0.4em\relax IEEE, 2021, pp.
  119--124.

\bibitem{martinez2019theory}
U.~Mart{\'\i}nez-Pe{\~n}as, ``{Theory of supports for linear codes endowed with
  the sum-rank metric},'' \emph{Designs, Codes and Cryptography}, vol.~87,
  no.~10, pp. 2295--2320, 2019.

\bibitem{camps2022optimal}
E.~Camps-Moreno, E.~Gorla, C.~Landolina, E.~L. Garc{\'\i}a,
  U.~Mart{\'\i}nez-Pe{\~n}as, and F.~Salizzoni, ``{Optimal Anticodes, {MSRD}
  Codes, and Generalized Weights in the Sum-Rank Metric},'' \emph{IEEE
  Transactions on Information Theory}, vol.~68, no.~6, pp. 3806--3822, 2022.

\bibitem{ott2022covering}
C.~Ott, H.~Liu, and A.~Wachter-Zeh, ``Covering properties of sum-rank metric
  codes,'' \emph{arXiv preprint arXiv:2210.02282}, 2022.

\bibitem{martinez2018skew}
``Skew and linearized reed--solomon codes and maximum sum rank distance codes
  over any division ring.''

\bibitem{alfarano2022sum}
G.~N. Alfarano, F.~J. Lobillo, A.~Neri, and A.~Wachter-Zeh, ``{Sum-rank product
  codes and bounds on the minimum distance},'' \emph{Finite Fields and Their
  Applications}, vol.~80, p. 102013, 2022.

\bibitem{martinez2020hamming}
U.~Mart{\'\i}nez-Pe{\~n}as, ``{Hamming and simplex codes for the sum-rank
  metric},'' \emph{Designs, Codes and Cryptography}, vol.~88, no.~8, pp.
  1521--1539, 2020.

\bibitem{almeida2020systematic}
P.~Almeida, U.~Mart{\'\i}nez-Pe{\~n}as, and D.~Napp, ``{Systematic maximum sum
  rank codes},'' \emph{Finite Fields and Their Applications}, vol.~65, p.
  101677, 2020.

\bibitem{martinezpenas2020sumrankBCH}
U.~Martínez-Peñas, ``{Sum-Rank {BCH} Codes and Cyclic-Skew-Cyclic Codes},''
  2020.

\bibitem{martinez2022generalMSRD}
U.~Mart{\'\i}nez-Pe{\~n}as, ``{A general family of {MSRD} codes and {PMDS}
  codes with smaller field sizes from extended Moore matrices},'' \emph{SIAM
  Journal on Discrete Mathematics}, vol.~36, no.~3, pp. 1868--1886, 2022.

\bibitem{caruso2022duals}
X.~Caruso and A.~Durand, ``{Duals of linearized Reed--Solomon codes},''
  \emph{Designs, Codes and Cryptography}, pp. 1--31, 2022.

\bibitem{boucher2020algorithm}
D.~Boucher, ``{An algorithm for decoding skew Reed--Solomon codes with respect
  to the skew metric},'' \emph{Designs, Codes and Cryptography}, vol.~88,
  no.~9, pp. 1991--2005, 2020.

\bibitem{puchinger2021bounds}
S.~Puchinger and J.~Rosenkilde, ``{Bounds on List Decoding of Linearized
  Reed-Solomon Codes},'' in \emph{2021 IEEE International Symposium on
  Information Theory (ISIT)}.\hskip 1em plus 0.5em minus 0.4em\relax IEEE,
  2021, pp. 154--159.

\bibitem{bartz2021fast}
H.~Bartz, T.~Jerkovits, S.~Puchinger, and J.~Rosenkilde, ``{Fast decoding of
  codes in the rank, subspace, and sum-rank metric},'' \emph{IEEE Transactions
  on Information Theory}, vol.~67, no.~8, pp. 5026--5050, 2021.

\bibitem{puchinger2022generic}
S.~Puchinger, J.~Renner, and J.~Rosenkilde, ``{Generic Decoding in the Sum-Rank
  Metric},'' \emph{IEEE Transactions on Information Theory}, 2022.

\bibitem{Hoermann2022efficient}
\BIBentryALTinterwordspacing
F.~H{\"o}rmann and H.~Bartz, ``{Efficient Decoding of Folded Linearized
  Reed-Solomon Codes in the Sum-Rank Metric},'' in \emph{WCC 2022: The Twelfth
  International Workshop on Coding and Cryptography}, March 2022. [Online].
  Available: \url{https://elib.dlr.de/146410/}
\BIBentrySTDinterwordspacing

\bibitem{Hoermann2022errorerasure}
\BIBentryALTinterwordspacing
F.~Hörmann, H.~Bartz, and S.~Puchinger, ``{Error-Erasure Decoding of
  Linearized Reed-Solomon Codes in the Sum-Rank Metric},'' 2022. [Online].
  Available: \url{https://arxiv.org/abs/2202.06758}
\BIBentrySTDinterwordspacing

\bibitem{jerkovits2022universal}
T.~Jerkovits, F.~H{\"o}rmann, and H.~Bartz, ``{Universal Decoding of
  Interleaved Linearized Reed--Solomon Codes in the Sum-Rank Metric},'' in
  \emph{Coding Theory and Cryptography: A conference in honor of Joachim
  Rosenthal’s 60th birthday}, 2022.

\bibitem{loidreau2006properties}
P.~Loidreau, ``{Properties of codes in rank metric},'' in \emph{{11th
  Inter-national Workshop on Algebraic and Combinatorial Coding Theory}}, 2008,
  pp. 192--198.

\bibitem{migler2004weight}
T.~Migler, K.~E. Morrison, and M.~Ogle, ``{Weight and rank of matrices over
  finite fields},'' \emph{arXiv preprint math/0403314}, 2004.

\bibitem{ratsaby2008estimate}
J.~Ratsaby, ``{Estimate of the number of restricted integer-partitions},''
  \emph{Applicable Analysis and Discrete Mathematics}, pp. 222--233, 2008.

\bibitem{gadouleau2008packing}
M.~Gadouleau and Z.~Yan, ``{Packing and covering properties of rank metric
  codes},'' \emph{IEEE Transactions on Information Theory}, vol.~54, no.~9, pp.
  3873--3883, 2008.

\bibitem{gadouleau2009bounds}
------, ``{Bounds on covering codes with the rank metric},'' \emph{IEEE
  Communications Letters}, vol.~13, no.~9, pp. 691--693, 2009.

\bibitem{gadouleau2008decoder}
------, ``On the decoder error probability of bounded rank-distance decoders
  for maximum rankdistance codes,'' \emph{IEEE Transactions on Information
  Theory}, vol.~54, no.~7, pp. 3202--3206, 2008.

\bibitem{delsarte1976properties}
P.~Delsarte, ``Properties and applications of the recurrence
  f(i+1,k+1,n+1)=q\^{}k+1f(i,k+1,n)-q\^{}kf(i,k,n),'' \emph{SIAM Journal on
  Applied Mathematics}, vol.~31, no.~2, pp. 262--270, 1976.

\end{thebibliography}
